\documentclass[10pt, conference]{ieeeconf}
\IEEEoverridecommandlockouts
\usepackage{enumerate}
\usepackage{array}
\usepackage{amsfonts,amssymb,amsxtra,balance}
\usepackage{mathrsfs}
\usepackage{calc,ifthen}
\usepackage{mathrsfs}
\usepackage{graphicx}
\usepackage{amssymb}

\newtheorem{assumption}{Assumption}

\newtheorem{proposition}{Proposition}
\newtheorem{corollary}{Corollary}

\newtheorem{lemma}{Lemma}
\newtheorem{remark}{Remark}

\newcommand{\Gstas}{G_2}

\newcommand{\col}{ \mbox{col} }

\def\L2e{{\cal L}_{2e}}

\def\rea{\mathbb{R}}

\def\diag{\mbox{diag}}

\def\rank{\mbox{rank}}

\def\begequarrs{\begin{eqnarray*}}
	\def\endequarrs{\end{eqnarray*}}
\def\begequarr{\begin{eqnarray}}
	\def\endequarr{\end{eqnarray}}
\def\begarr{\begin{array}}
	\def\endarr{\end{array}}
\def\begequ{\begin{equation}}
	\def\endequ{\end{equation}}
\def\lab{\label}
\def\begdes{\begin{description}}
	\def\enddes{\end{description}}
\def\begenu{\begin{enumerate}}
	\def\begite{\begin{itemize}}
		\def\endite{\end{itemize}}
	\def\endenu{\end{enumerate}}

\def\lef[{\left[\begin{array}}
	\def\rig]{\end{array}\right]}

\def\begcen{\begin{center}}
	\def\endcen{\end{center}}
\def\begrem{\begin{remark}\rm}
	\def\endrem{\end{remark}}
\def\begcas{\begin{cases}}
	\def\endcas{\end{cases}}

\title{\LARGE \bf {Robust PI Passivity--based Control of Nonlinear Systems: Application to Port--Hamiltonian Systems and Temperature Regulation}}%
\date{}
\author{S. Aranovskiy$^{1}$ R. Ortega$^{2}$ and R. Cisneros$^{2}$
\thanks{$^{1}$S. Aranovskiy is with the Department of Control Systems And Informatics, Saint-Petersburg University ITMO :  {\tt\small s.aranovskiy@gmail.com}}%
\thanks{$^{2}$R. Ortega and R. Cisneros are with the LSS-Supelec, 3, Rue Joliot-Curie, 91192 Gif--sur--Yvette, France : {\tt\small ortega@lss.supelec.fr} and {\tt\small rafael.cisneros@lss.supelec.fr} }}
\begin{document}
\maketitle
\thispagestyle{empty}
\pagestyle{empty}
	
\begin{abstract}
This paper deals with the problem of control of partially known nonlinear systems, which have an open--loop stable equilibrium, but we would like to add a {\em PI controller} to regulate its behavior around another operating point. Our main contribution is the identification of a class of systems for which  a globally stable PI can be designed knowing  {\em only} the systems input matrix and measuring {\em only} the actuated coordinates. The construction of the PI is done invoking passivity theory. The difficulties encountered in the design of {\em adaptive} PI controllers with the existing theoretical tools are also discussed. As an illustration of the theory, we consider port--Hamiltonian systems and a class of thermal processes.
\end{abstract}

%
\section{INTRODUCTION}
\lab{sec1}	
%
In many practical applications the plant to be controlled has an open--loop stable equilibrium, {\em e.g.}, at the origin, and we would like to add a controller to regulate its behavior around another operating point. In the case of linear systems the dynamics remains invariant under coordinate shifts, therefore this task can be easily accomplished using the incremental model of the plant. Unfortunately, this is not the case for nonlinear systems, for which there is no obvious advantage of working with the incremental model.

Another common requirement in applications is the use of standard proportional-integral (PI) controllers, which overwhelmingly dominate the industrial market. Although commissioning a PI to operate around a single operating point is relatively easy, the performance will be below par in wide operating regimes, which is the scenario in modern high--performance applications. To overcome this drawback the current practice is to re--tune the gains of the PI controllers based on a linear model of the plant evaluated at various operating points, a procedure known as gain--scheduling. There are several  disadvantages of gain--scheduling including the need to switch (or interpolate) the controller gains  and the non--trivial definition of the regions in the plants state space where the switching takes place---both problems are exacerbated if the dynamics of the plant is highly nonlinear. Another common commissioning procedure is to use auto--tuners,  that heavily rely on the availability of a ``good" linear approximation of the plant dynamics. To avoid the need to rely on linearization it is necessary to develop a procedure to design robust PI controllers for nonlinear systems with uncertain parameters.

Motivated by the discussion above in this paper we identify a class of (input affine) nonlinear systems for which it is possible to design a PI controller with the following features.
\begite
\item[F1] Regulation of the closed--loop system around the desired (non--zero) operating point should be guaranteed.
\item[F2] The PI controller should be {\em robust}, in the sense that reduced knowledge of the system parameters is required.
\item[F3] To simplify the controllers commissioning, a well defined admissible range of variation for the PI proportional and integral gains, preserving closed--loop stability, should be provided.
\endite
We propose the construction of a PI controller with the features F1--F3 for plants with unknown dynamics verifying the following assumptions.
\begite
\item[A1] The open--loop system is {\em unknown} but has a stable equilibrium at the origin.
\item[A2] The desired equilibrium belongs to the assignable set and admits a {\em convex} Lyapunov function.
\item[A3] The Lyapunov function is the sum of two functions, depending on the un--actuated and actuated coordinates, respectively. The first function is {\em unknown} while the second one is separable and linearly parameterized in terms of some {\em unknown parameters}.
\item[A4] The input matrix is constant, known and has $n-m$ zero rows, where $n$ and $m$ are the dimensions of the state  and  input vectors, respectively.
\endite
As indicated in the article's title we exploit the fundamental property of {\em passivity} to design the proportional-integral passivity based controller, which will be referred in the sequel as PI--PBC. The first step in the design is to, relying on A1 above, invoke the celebrated theorem of Hill and Moylan \cite{VAN} to identify a suitable passive output for the system,  with storage function the Lyapunov function of the open--loop system. Since our interest is the regulation of non--zero equilibria, we then use the results of \cite{JAYetal} to create a new passive output for the incremental model with a storage function that has a minimum at the desired equilibrium. As shown in \cite{JAYetal}, feeding back the passive output through a PI controller ensures stability of the desired equilibrium {\em for all} positive definite PI gains. It is important to underscore that, since the passivity property has been established for the incremental model, the equilibrium can also be stabilized setting the control input equal to the (constant) value that assigns the equilibrium, say $u^*$, which is univocally defined. However, this open--loop control will, obviously, be non--robust. In the robustness context of the present paper neither the plant dynamics nor the Lyapunov function are known and, consequently, we cannot compute neither the passive output nor $u^*$. It is at this point that we invoke A3 and A4 above to prove that, under these assumptions, it is possible to define ranges for the proportional and integral gains that make the PI--PBC implementable and, consequently, guarantee stability of the equilibrium. Another important feature of the proposed PI--PBC is that it requires only partial measurement of the state, namely, only the $m$ state variables associated to the non--zero rows of the input matrix, referred in the sequel as {\em actuated} coordinates.

Several practical applications of PI--PBC have been reported in the literature. This include, RLC circuits \cite{CASetal}, power converters \cite{HERetal}, fuel cells \cite{TALetal}, electric drives \cite{MARALE} and mechanical systems \cite{SANetal}. In \cite{DONJUN} a procedure to add an integral action to a non--passive output for a class of port--Hamiltonian systems was first proposed, and later extended in \cite{ORTROM}, \cite{ROMDONORT}. To the best of our knowledge, the present paper is the first attempt to design PI--PBCs with guaranteed stability properties for systems with partially known dynamics.

A natural question that arises at this point is the incorporation of adaptation in the design of the PI (or PID). In the  power converter application of \cite{HERetal} a parameter that enters in the definition of the passive output, {\em i.e.}, the load resistance, is adaptively identified---however, all other parameters are assumed to be known. In the interesting paper \cite{ANTAST} it is shown that it is possible to adaptively estimate $u^*$ for a general nonlinear system with scalar input, keeping the estimate in a known interval, provided the passive output is known. In spite of a large number of publications the problem of designing a {\em provably stable} adaptive PID for systems with unknown parameters remains, as far as we know, open. The difficulty of this task was identified already in 1984 in \cite{ORTKEL}. As is well--known  \cite{SASBOD}, the stability of indirect adaptive methods relies on parameter convergence that, in its turn, requires persistency of excitation---a property that is not satisfied in the regulation tasks where PI control is effective. On the other hand, the application of direct methods is stymied by the absence of a suitable parameterization of this structure--constrained controller. For the PI--PBC studied in this paper the main difficulty is the need to estimate two objects, that appear multiplicatively in the Lyapunov analysis: the passive output and the ideal control signal  $u^*$. This point is further elaborated in Subsection \ref{subsec52}.

The remaining of the paper is organized as follows. Section \ref{sec2} presents the problem formulation. To streamline the presentation of the main result, which is given in  Section \ref{sec4}, some preliminary lemmata are given in Section \ref{sec3}. In Section \ref{sec5} we discuss the reasons that stymie the use of adaptation and the inability to state a robustness result based on continuity and approximate prior knowledge of the plant. Section \ref{sec6} is devoted to application of the proposed PI--PBC for port--Hamiltonian (pH) systems \cite{VAN}  and a temperature regulation problem. The paper is wrapped--up with concluding remarks in  Section \ref{sec7}.
\vspace{2mm}

\noindent {\bf Notations} $I_n$ is the $n \times n$ identity matrix and $\mathbf{0}_{n \times s}$ is an
$n \times s$ matrix of zeros, $\mathbf{0}_n$ is an $n$--dimensional column vector of zeros. Given $a_i \in \rea,\; i \in \bar n := \{1,\dots,n\}$, we denote with $\col(a_i)$ the $n$--dimensional column vector and $\diag\{a_i\}$ the diagonal $n \times n$ matrix with
elements $a_i$. For $x \in \rea^n$,  $|x|$ is the Euclidean norm. For mappings of scalar argument $g:\rea \to \rea^s$, $g'$ and $g''$ denote first and second order differentiation, respectively. For mappings $f:\rea^n \to \rea$, $\nabla f :=
(\frac{\partial f}{\partial x})^\top$ and $\nabla^2 f:= \frac{\partial^2 f}{\partial x^2}$. For the distinguished element $x^* \in \rea^n$ and any mapping $F:\rea^n \to \rea^s$ we denote $F^*:=F(x^*)$ and the error signal  $\tilde F(x):=F(x)-F^*$. 
%
\section{PROBLEM FORMULATION}
\lab{sec2}	
%
In this section we formulate the control problem addressed in the paper, enunciate the assumptions made on the plant to solve it and make some remarks on these assumptions.
%
\subsection{Robust PI control problem}
\lab{subsec21}	
Consider the nonlinear, input affine, system
\begin{align}
	\dot x &= f(x) + G u, \label{systform}
\end{align}
where $x\in \mathbb{R}^n$, $u\in\mathbb{R}^m$, $n >m$, $f: \mathbb{R}^n \to \mathbb{R}^n$ is an {\em unknown} smooth mapping, $G\in\mathbb{R}^{n\times m}$ is {\em constant} verifying $\rank(G)= m$.

The following is a key assumption made throughout the paper.

\begin{assumption} 
\label{as:g}
The matrix $G$ has $n-m$ zero rows. Without lost of generality\footnote{See R6 in the next subsection and Subsection \ref{subsec51} for more general forms of $G$.} it is assumed of the form
\begequ
\lab{g}
	G=\begin{bmatrix} \mathbf{0}_{(n-m)\times m} \\ G_2\end{bmatrix},
\endequ
where $G_2 \in \mathbb{R}^{m \times m}$ is {\em known}.
\end{assumption}

This assumption can be easily {\em obviated} introducing state and input changes of coordinates. Indeed, it is well--known---see, {\em e.g.}, Theorem 2 of Section 2.7 of \cite{LANTIS}---that for any full rank, matrix $G\in\mathbb{R}^{n\times m}$ there exists (elementary) full rank matrices $T \in \rea^{n \times n}$ and  $S \in \rea^{m \times m}$ such that
$$
T G S = \begin{bmatrix} \mathbf{0}_{(n-m)\times m} \\ I_m\end{bmatrix}.
$$
Consequently, introducing $z = Tx$ and $v=S^{-1}u$ the system \eqref{systform} takes the desired form
$$
\dot z = w(z) + \begin{bmatrix} \mathbf{0}_{(n-m)\times m} \\ I_m\end{bmatrix}v,
$$
where $w(z)=Tf(T^{-1}z)$. We should note, however, that a change of state representation destroys---in general---the {\em original structure} of the system, whose knowledge may be critical for the verification of the second assumption below. This fact is clearly illustrated in the physical examples considered in Section \ref{sec6}. For this reason, we prefer to leave it as an standing assumption. 

Motivated by Assumption \ref{as:g} we find convenient to define a partition of the state vector into its un--actuated and actuated components as
$$
x=\lef[{c} x_u \\ x_a\rig],\;x_u:=\lef[{c} x_1 \\ x_2 \\ \vdots \\ x_{n-m}\rig],\;x_a:=\lef[{c} x_{n-m+1}\\x_{n-m+2}\\ \vdots \\ x_n\rig].
$$
It is assumed that {\em only $x_a$ is available for measurement.}

The unforced system, that is, $\dot x = f(x)$,
has a stable equilibrium at the origin with a {\em partially known} Lyapunov function.  We are interested in controlling the system with a PI at a non--zero equilibrium---a situation that arises in most practical applications.
Thus, we are given a desired equilibrium point, $x^{\star}\in\mathbb{R}^n$, and the control goal is to ensure {\em stability} of this equilibrium using a PI control law of the form
\begequarrs
\dot{z} &= - K_I \psi(x_a,x^*)\\
u &= - K_P  \psi(x_a,x^*) + z,		
\endequarrs
where $z \in \rea^m$ is the controller state, $K_P \in \rea^{m \times m}$ and $K_I \in \rea^{m \times m}$ are tuning gains and $\psi: \mathbb{R}^m \times \mathbb{R}^n \to \mathbb{R}^m$ is a mapping designed with the partial knowledge of the aforementioned Lyapunov function.

The following, practically reasonable, assumption is made throughout the paper.
\begin{assumption} 
\label{as:xstar}
The desired equilibrium point $x^\star$ belongs to the assignable equilibrium set, that is,
\begin{equation}\label{eqset}
	x^\star\in\mathcal{E}:=\left\{x \in \rea^n \ |\  \lef[{ccc} I_{n-m} & | & \mathbf{0}_{(n-m) \times n} \rig] f(x)= 0 \right\}.
\end{equation}
\end{assumption}
%
\subsection{Assumptions on the open--loop plant}
\lab{subsec22}
The following assumption identifies the class of vector fields $f(x)$ for which we provide an answer to the problem.\\

\begin{assumption} 
\label{as:H}
For the system \eqref{systform} there exists a twice--differentiable, positive definite function $H:\mathbb{R}^n \to \mathbb{R}_{\geq 0}$, verifying the following.\\
\begin{enumerate}[(i)]
	\item $[\nabla H(x)] ^\top f(x) \le 0$.\\
	\item $[\nabla H (x)-\nabla H(x^\star)]^\top \tilde f(x)=:-Q(x) \le 0$.\\
	\item The function $H(x)$ is of the form
$$
H(x)=H_u(x_u) + H_a(x_a)
$$
with
\begequ
\lab{h}
H_a(x_a) = \sum_{i=n-m+1}^n{d_i \phi_i(x_i)},
\endequ
where the function $H_u:\rea^{n-m} \to \rea$ and the constants $d_i>0$ are {\em unknown} but the  the functions  $\phi_i: \mathbb{R} \to \mathbb{R}$ are {\em known}.\\
\item The functions $H_u(x_u)$ and  $\phi_i(x_i)$ are {\em convex}.
\end{enumerate}
\end{assumption}
%
\subsection{Discussion}
\lab{subsec23}
The following remarks regarding the assumptions are in order.

\begite
\item[R1] Although the set ${\cal E}$ is not known, it is reasonable to assume that we have enough prior knowledge about the plant to select the desired operating point as a feasible equilibrium. Hence, Assumption \ref{as:xstar} is reasonable.

\item[R2] A corollary of Assumption \ref{as:xstar}  is that  the constant input $u^\star$, that assigns the equilibrium, is uniquely defined as
\begin{equation}\label{ustar}
	u^\star:=\left(G_2^\top G_2 \right)^{-1}\begin{bmatrix} \mathbf{0}_{m\times (n-m)} & G^\top_2\end{bmatrix} f^\star.
\end{equation}
Notice that, without knowledge of $f(x)$, this constant cannot be computed.

\item[R3] Since the open--loop system \eqref{systform} has a stable equilibrium at the origin Assumption \ref{as:H} (i) follows as a corollary of Lyapunov's converse theorems \cite{KHA}. As will become clear below Assumption  \ref{as:H} (ii) and (iv) are required to prove passivity of the incremental model as done in \cite{JAYetal}. 

\item[R4] We underscore that no assumption, beyond twice differentiability and convexity, is imposed on the unknown component $H_u(x_u)$ of the  Lyapunov function of the open--loop system $H(x)$. On the other hand, stricter conditions are imposed on the second component $H_a(x_a)$, with uncertainty captured by the unknown constants $d_i$.

\item[R5] Assumptions  \ref{as:H} (iii) and Assumption  \ref{as:g} are the key requirements imposed on the plant to design the robust PI--PBC. This assumption is satisfied by a large class of physical systems, including a class of port--Hamiltonian \cite{VAN}  and thermal systems studied in Section \ref{sec6}.

\item[R6] Regarding Assumptions  \ref{as:g}, in the more general case when $G$ is not of the form \eqref{g} an additional shuffling of the rows of $G$ is needed in the design. This procedure is explained in  Subsection \ref{subsec51}.

\item[R7] For quadratic Lyapunov functions of the form $H(x)=x^\top P x$, with $P>0$, Assumption  \ref{as:H} (ii) is satisfied if the open--loop system is {\em convergent} in the sense of Demidovich \cite{PAVetal}. That is, if it satisfies
$$
P \nabla f(x) +  [\nabla f(x)]^\top P \leq 0.
$$  
\endite
%
\section{PRELIMINARY LEMMATA}
\lab{sec3}	
%
Unless otherwise indicated, throughout the rest of the paper Assumption \ref{as:g} holds. Define for the system \eqref{systform} the output
\begin{equation} \label{eq:youtput}
	y=G^\top \nabla H(x) = G_2^\top D \Phi(x_a),
\end{equation}
where 
\begequarrs
D & := & \lef[{cccc} d_{n-m+1} & 0 & \dots & 0 \\ 0 & d_{n-m+2} & \dots & 0 \\ \vdots & \vdots & \vdots & \vdots \\ 0 & 0 & \dots & d_n \rig]>0 \\
\Phi(x_a) & :=  & \lef[{c} \phi'_{n-m+1}(x_{n-m+1}) \\ \vdots \\ \phi'(x_n) \rig].
\endequarrs
A corollary of the theorem of Hill and Moylan \cite{VAN} is that, if Assumption \ref{as:H} (i) holds, the system \eqref{systform}, \eqref{eq:youtput} defines a passive mapping $u \mapsto y$ with storage function $H(x)$.

To operate the system at a non--zero equilibrium it is necessary to shift the minimum of the storage function and define the passivity property between the incremental input and the output error. Towards this end, we
recall Proposition 1 of \cite{JAYetal} and state it as a lemma below.
\begin{lemma} 
\lab{lem1}
Consider the incremental model of the system \eqref{systform}, \eqref{eq:youtput}
\begin{equation} \label{eq:IncrModel}
	\begin{aligned}
		\dot x &= f(x) + G u^{\star} + G \tilde{u},  \\
		e &= G_2^\top D \tilde {\Phi}(x_a),
	\end{aligned}
\end{equation}
where $\tilde{u} = u - u^{\star}$ is the incremental input. Under Assumptions \ref{as:g}--\ref{as:H} the mapping $\tilde u \mapsto e$  is passive with storage function $U: \mathbb{R}^n\to \mathbb{R}_{\geq 0}$ given by
\begin{equation}
\lab{h0}
			U(x)=H(x)-x_u^\top \nabla H_u^*- x_a^\top D\Phi^\star+k,
\end{equation}
where $k$ is a constant that ensures $U(0)=0$. More precisely,
\begequ
\lab{doth0}
\dot{U} = -Q(x) + e^\top \tilde{u},
\endequ
where $Q(x)$ is defined in Assumption \ref{as:H} (ii).
\end{lemma}

One of the main interests of passive systems is that they can be globally stabilized with PI controls (with arbitrary positive definite gains). This well--known fact is stated in the lemma below, whose proof is given to streamline the presentation of our main result.

\begin{lemma} 
\lab{lem2}
Consider the system \eqref{systform} verifying Assumptions \ref{as:g}--\ref{as:H} in closed--loop with the PI--PBC
\begequ
\lab{u}
\begin{aligned}
                 e&=G_2^\top D \widetilde {\Phi}(x_a)\\
                 \dot{z} &= - K_I e\\
		u &= - K_P e + z.
	\end{aligned}
\end{equation}
{\em For all} positive definite gain matrices $K_P \in \rea^{m \times m}$ and $K_I \in \rea^{m \times m}$ all trajectories are bounded, the equilibrium point $(x,z)=(x^*,u^*)$ is {\em globally stable} (in the sense of Lyapunov) and the augmented error signal
\begequ
\lab{ea}
e_a:=\lef[{c} Q(x) \\ e \rig],
\endequ
where $Q(x)$ is defined in Assumption \ref{as:H} (ii), verifies
\begequ
\lab{limea}
\lim_{t \to \infty} e_a(t)=0.
\endequ
Moreover, if $e_a$ is a detectable output for the closed--loop system then the equilibrium point is {\em asymptotically} stable.
\end{lemma}
\begin{proof}
Defining $\tilde z:=z-u^*$ the last two equations of the controller \eqref{u} may be written in the form
\begequ
\lab{u1}
\begin{aligned}
                 \dot {\tilde z} &= - K_I e\\
		\tilde u &= - K_P e + \tilde z.
	\end{aligned}
\end{equation}
Consider the Lyapunov function candidate
\begequ
\lab{w}
	W(\tilde{z},x) = U(x) + \frac{1}{2}\tilde{z}^\top \Lambda_I\tilde{z},
\endequ
where $\Lambda_I>0$. The time derivative of the Lyapunov function along the trajectories of the closed--loop system is
\begin{equation} \label{dotw}
	\begin{aligned}
		\dot{W} &= -Q(x) + e^\top \tilde{u} + \tilde{z}^\top \Lambda_I \dot{\tilde{z}} \\
		&= -Q(x) - e^\top K_P e+ \tilde{z}^\top e - \tilde{z}^\top \Lambda_I K_I e.
	\end{aligned}
\end{equation}
Setting $ \Lambda_I= K_I^{-1}$ yields
$$
\dot{W} = -Q(x) - e^\top K_P e.
$$
The proof is complete invoking standard Lyapunov arguments \cite{KHA}.
\end{proof}
%
\section{THE ROBUST PI--PBC}
\lab{sec4}	
%
As indicated in R4 of Subsection \ref{subsec23} the matrix $D$ is unknown. Hence, the error signal $e$ cannot be constructed and the PI--PBC  \eqref{u} is not implementable.	
This motivates our main result given below.
\begin{proposition} \lab{pro1}
Consider system \eqref{systform} verifying Assumptions 1--3 in closed--loop with the robust PI--PBC
\begin{equation}\lab{unew}
	\begin{aligned}
		u &= - K_P  \widetilde{\Phi}(x_a) + z \\
		\dot{z} &= - K_I  \widetilde{\Phi}(x_a),
	\end{aligned}
\end{equation}
with the controller gains
\begequarr
\nonumber
K_P & = & G_2^{-1} \Gamma_P \\
K_I & = & G_2^{-1} \Gamma_I.
\lab{gaipi}
\endequarr
{\em For all} diagonal, positive definite matrices  $\Gamma_P \in \mathbb{R}^{m \times m}$ and $\Gamma_I \in \mathbb{R}^{m \times m}$  we have the following.
\begite
\item[(i)] All trajectories are bounded and the equilibrium point $(x,z)=(x^*,u^*)$ is {\em globally stable} (in the sense of Lyapunov).
\item[(ii)] The augmented error signal $e_a$ defined in \eqref{ea} verifies \eqref{limea}.  
\item[(iii)] If $e_a$ is a detectable output for the closed--loop system then the equilibrium point is globally {\em asymptotically} stable.
\endite
\end{proposition}
\begin{proof}
Some simple manipulations prove that
\begin{equation} \label{lamp}
K_P \widetilde{\Phi}(x_a) =  G_2^{-1} \Gamma_P D^{-1} G_2^{-\top}  G_2^\top D  \widetilde{\Phi}(x_a)= \Lambda_P e,
\end{equation}
where we defined the matrix
\begequ
\lab{lamp1}
\Lambda_P := G_2^{-1} \Gamma_P D^{-1} G_2^{-\top},
\endequ
and used the definition of $e$ in \eqref{u}. Invoking Sylvester's Law of Inertia \cite{LANTIS}, and the fact that $\Gamma_P$ and $D$ are diagonal and positive definite, we have that $\Lambda_P>0$. 

Next choose
\begequ
\lab{lami1}
\Lambda_I := G_2^\top D \Gamma_I^{-1} G_2,
\endequ
that is, also, positive definite for all diagonal, positive definite $\Gamma_I$. Then
\begin{equation} \label{lami}
	\Lambda_I K_I \widetilde{\Phi}(x) = G_2^\top D \widetilde{\Phi}(x_a) =e.
\end{equation}
Replacing \eqref{lamp} and \eqref{lami} in the controller equations yields
\begin{equation*}
	\begin{aligned}
		\tilde u &= - \Lambda_P e + \tilde z \\
		\dot{\tilde z} &= - \Lambda_I^{-1} e.
	\end{aligned}
\end{equation*}
Consequently, the time derivative of the Lyapunov function \eqref{dotw} becomes now
\begequ
\lab{dotw1}
	\dot{W} = -Q(x) - e^\top \Lambda_P e,
\endequ
completing the proof.
\end{proof}

To obtain an implementable version of the robust PI--PBC it was necessary to carry--out two tasks. First, to make the damping injection, introduced by the proportional term, function of the unknown matrix $D$. Indeed, replacing \eqref{lamp1} in \eqref{lamp} we get
$$
K_P \widetilde{\Phi}(x)=G_2^{-1} \Gamma_P D^{-1} G_2^{-\top}e.
$$
Second, make the gain $\Lambda_I$ of the Lyapunov function \eqref{w} also a function of $D$---see \eqref{lami1}. 

An important observation is that, even though the controller only requires measurement of the actuated terms of the state $x_a$,  it achieves regulation of the full state vector.
%
\section{ADDITIONAL REMARKS ON THE PI--PBC}
\lab{sec5}	
%
%
In this section we explain how to proceed when $G$ is not of the form \eqref{g}, discuss the reasons that stymie the use of adaptation and the inability to state a robustness result based on continuity and approximate prior knowledge of the matrix $D$.
\subsection{General $G$ (with $n-m$ zero rows)}
\lab{subsec51}	

Instrumental to design the robust PI--PBC was the particular form of $H(x)$ defined in Assumption \ref{as:H} (iii). In view of the construction of the robust PI--PBC, it is clear that if $G$ is not of the form  \eqref{g} the assumption must be modified redefining the actuated and un--actuated coordinates. 

To avoid cluttering the notation we will explain the procedure only for the case when $n=3$ and $m=2$---the general case follows {\em verbatim}. Assume, furthermore, that $G$ is of the form
$$
G=\lef[{c} g_1^\top \\ \mathbf{0}_{1 \times 2} \\  g_3^\top \rig].
$$
The form of $H(x)$ given in Assumption \ref{as:H} (iii) must be, accordingly, modified to
$$
H(x)=H_u(x_2)+d_1 \phi(x_1) + d_3 \phi(x_3).
$$
In this case the passive output $e$ for the incremental model becomes
$$
G^\top [\nabla H(x)-\nabla H(x^*)]=G_s \lef[{cc} d_1 & 0 \\ 0 & d_3 \rig] \lef[{c} \widetilde{\Phi}_1(x_1) \\  \widetilde{\Phi}_3(x_3) \rig].
$$
where
$$
G_s:= \lef[{ccc} g_1& | & g_3 \rig].
$$
The robust PI--PBC is given by
\begin{equation*}
	\begin{aligned}
		u &= - G_s^{-1} \Gamma_P \lef[{c} \widetilde{\Phi}_1(x_1) \\  \widetilde{\Phi}_3(x_3) \rig] + z \\
		\dot{z} &=  - G_s^{-1} \Gamma_I \lef[{c} \widetilde{\Phi}_1(x_1) \\  \widetilde{\Phi}_3(x_3) \rig],
	\end{aligned}
\end{equation*}
where $\Gamma_P$ and $\Gamma_I$ are arbitrary, diagonal, positive definite matrices.

Before closing this subsection we remark that our construction {\em critically} relies on the assumption of existence of $n-m$ zero rows in $G$. Indeed, it is possible to show that if this is not the case, even assuming $H(x)$ of the form
$$
H(x)=\sum_{i=1}^n d_i \phi_i(x_i)
$$
and defining $D_n = \diag\{d_i\}$, it is not possible to find an $m \times m$ positive definite matrix $\Lambda$, which will depend on $D_n$, such that the matrix $\Lambda G^\top D_n$ is {\em independent} of $D_n$. The fact that this is {\em not possible} for all matrices $G$ is obvious considering the counterexample $G=\col(1,1)$. Hence, the assumption of existence of $n-m$ zero rows in $G$ is {\em necessary} to solve the problem. 
%
\subsection{Difficulties for adaptation}
\lab{subsec52}	
A natural alternative to the robust PI--PBC presented above is to assume a parametrisation of $f(x)$ and try to estimate this parameters or, in a direct approach, estimate the matrix $D$ that defines the passive output. The indirect approach, as is well--known,  relies on parameter convergence that requires persistency of excitation---a property that is not satisfied in the regulation tasks where PI control is effective.

Let us see what are the difficulties for the application of a direct adaptation approach. Towards this end, we propose the adaptive PI--PBC
$$
\begin{aligned}
                 \dot{\hat D} &=F(x,z)\\
                 \hat e&=\Gstas^\top  \hat D\; \widetilde {\Phi}(x_a)\\
                 \dot{z} &= - K_I \hat e\\
                   u &= - K_P \hat e + z,
	\end{aligned}
$$
where the parameter adaptation law $F:\rea^n \times \rea^m \to \rea^{m \times m}$ is to be defined.\footnote{Notice that, in contrast to the robust PI--PBC, we have assumed that the full state is measurable.} Defining $\tilde e:=\hat e-e$ the last two equations of the controller may be written in the form
$$
\begin{aligned}
                 \dot {\tilde z} &= - K_I( e+\tilde e)\\
		\tilde u &= - K_P (e+\tilde e) + \tilde z.
	\end{aligned}
$$
The time derivative of the Lyapunov function \eqref{w} with $ \Lambda_I= K_I^{-1}$ is now
$$
	\begin{aligned}
		\dot{W} &= -Q(x) - \hat e^\top K_P \hat e - \tilde{u}^\top \tilde e\\
		 &= -Q(x) - \hat e^\top K_P \hat e - \tilde{u}^\top \Gstas^\top  \tilde D \widetilde {\Phi}(x_a)
	\end{aligned}
$$
where we underscore the presence of the last right hand term. If $\tilde u$ were known the standard procedure of augmenting the Lyapunov function with a term $\mbox{trace} (\tilde D^\top \tilde D)$ and cancelling the sign--indefinite term with a suitable choice of $F(x,z)$ would do the job. Alas, $u^*$ is not known, hence this approach is not feasible. 

Adding an adaptation for the constant $u^*$ is also not a trivial task, because of the bilinear nature of the joint estimation problem.
%
\subsection{Comments on robustness based on continuity}
\lab{subsec53}	
The availability of a {\em bona fide} Lyapunov function for the known parameters PI--PBC, {\em i.e.}, $W(x,\tilde z)$, suggests that stability will be preserved if the matrix $D$ is replaced by a ``good", {\em constant} estimate of it, say $D_0$. More precisely, it is expected that replacing the controller \eqref{u} by
$$
\begin{aligned}
                 e_0 &=\Gstas^\top D_0 \widetilde {\Phi}(x_a)\\
            \dot{z} &= - K_I  e_0\\
	      		u &= - K_P e_0 + z,
	\end{aligned}
$$
where
$$
D=D_0 + \Delta,\;\Delta:=\diag\{\delta_i\}
$$
would ensure stability if $|\col(\delta_i)|$ is sufficiently small. Unfortunately, since the Lyapunov function is {\em not strict}, this conjecture cannot be analytically validated. Indeed, in this case the time derivative of the Lyapunov function \eqref{w} with $ \Lambda_I= K_I^{-1}$ is now
$$
	\begin{aligned}
		\dot{W} &= -Q(x) + e^\top \tilde u - \tilde{z}^\top (e-\Gstas^\top \Delta \widetilde {\Phi}(x_a))\\
		 &= -Q(x) - e_0^\top K_P e_0 - (K_P e_0 - \tilde z)^\top \Gstas^\top  \Delta \widetilde \Phi(x_a).
	\end{aligned}
$$
While the term $e_0^\top K_P \Gstas^\top  \Delta \widetilde \Phi(x_a)$ can be dominated for ``small" $\Delta$, there is no way we can dominate the remaining term $ \tilde z^\top K_I \Gstas^\top  \Delta \widetilde \Phi(x_a)$ and the Lyapunov analysis cannot be completed with standard techniques.

This unfortunate situation does not mean, of course, that a continuity result of this type cannot be established. It simply reveals our inability to do it with the tools used to analyze the ideal case.
%
\section{EXAMPLES}
\lab{sec6}	
%
Two examples of physical systems  which are amenable for robust PI--PBC are given below. Attention is concentrated on the verification of Assumption \ref{as:H}. Hence, unless otherwise indicated, Assumption \ref{as:g} is not imposed.
\subsection{A class of port--Hamiltonian system}
\lab{subsec61}	
%
\begin{proposition}
The pH system
\begin{equation}\label{phsys}
\dot{x}  =   ({\cal J} -{\cal R} ) \nabla H(x) + Gu
\end{equation}
with {\em constant} interconnection  ${\cal J}=-{\cal J}^\top$ and damping ${\cal R}={\cal R}^\top \geq 0$ matrices satisfies Assumption \ref{as:H} (i) and (ii).
\end{proposition}
\begin{proof}
Assumption  \ref{as:H}  (i) is obviously satisfied because
$$
 \nabla^\top H(x) ({\cal J} -{\cal R} )\nabla H(x) =- \nabla^\top H(x) {\cal R} \nabla H(x) \leq 0.
$$
Similarly, Assumption  \ref{as:H}  (ii) holds since
\begequarrs
&&[ \nabla H(x)- \nabla H(x^*)]^\top  ({\cal J} -{\cal R} )[\nabla H(x)- \nabla H(x^*)] =\\
&& -[ \nabla H(x)- \nabla H(x^*)]^\top {\cal R} [\nabla H(x)- \nabla H(x^*)]\leq 0.
\endequarrs
\end{proof}
\subsection{Application to temperature regulation}
\lab{subsec62}	
%
In this subsection we design a robust PI--PBC for the temperature regulation of a class of thermal systems---the so--called, rapid thermal processes.
\subsubsection{System Description}
%
Similarly to \cite{ebert2004model, schaper1992modeling} we consider the following model of rapid thermal processes
\begin{equation} \label{tsys}
	\begin{aligned}
		\dot{T} = A_1 \left[\Psi(T)-\Psi(T_{rad})\right]+A_2 \left(T-T_{conv} \right) +Gu,   
	\end{aligned}
\end{equation} 
where $T\in \mathbb{R}^n_{\geq 0}$  represents the vector of temperatures, ${\Psi(T):=\col(T_i^4)}$ and  $T_{rad},\;T_{conv} \in \mathbb{R}^n_{\geq 0}$  are, respectively, the vectors of temperatures related to the radiation heat emission from environment and the convection air flows. 
The constant matrices  $A_1,A_2\in\mathbb{R}^{n\times n}$ are {\em Hurwitz}  and contain the radiation and the convection heat transfer coefficients. Also, the entries of  $G\in\mathbb{R}^{n\times m}$ correspond to the heat transfer coefficients of the heating elements. Finally, $u\in \mathbb{R}^m$ is the controlled power applied to the heating elements. 

In the model above, as in \cite{schaper1992modeling}, it is considered that  \eqref{tsys}  is heated almost uniformly so that the contribution of energy from conduction is too small with respect to the radiation transfer. Hence, the conduction heat transfer is neglected. 

To simplify the notation we re--write the system \eqref{tsys} in the form 
$$
\dot{T} = A_1 \Psi(T)+A_2 T+E +Gu,
$$
where
$$
E:= -  A_1 \Psi(T_{rad}) - A_2 T_{conv}. 
$$
Unlike  $A_1,A_2$ and $E$ that are highly uncertain, the input matrix $G$---that is defined by the induced heat flow---can be precisely identified. The {\em control objective} is then to design a robust PI, {\em i.e.}, that does not require the knowledge of  the uncertain parameters, to regulate the process around some desired temperature, which is {\em not equal} to the open--loop equilibrium, but belongs to the assignable equilibrium set, that is,
\begequ
\lab{tsta}
T^\star \in \left\{T \in \rea^n_{\geq 0}\ |\ G^\perp [A_1 \Psi(T)+A_2 T + E] = 0 \right\},
\endequ
where $G^\perp \in \rea^{(n-m) \times n}$ is a full-rank left-annihilator of $G$.

To place the problem in the context of Proposition 1 we first shift the equilibrium of the open--loop system to the origin and then proceed to verify Assumption \ref{as:H}. For, we  introduce the standard change of coordinates
$$
x=T-\bar T,
$$
where $\bar T$ is the open--loop equilibrium that satisfies
\begequ
\lab{opelooequ}
A_1 \Psi(\bar T)+A_2 \bar T+E=0.
\endequ
Thus, the system \eqref{tsys} in the new coordinates takes the form \eqref{systform} with
\begin{equation}
\label{vecfie}
	f(x):=  A_1 \Psi(x+\bar T)+A_2(x+\bar T)+E,
\end{equation}
Associated to the desired temperature $T^\star$ we define the equilibrium to be stabilised
\begequ
\lab{xsta}
x^\star := T^\star - \bar T.
\endequ
\subsubsection{Passivity of the thermal system}

The lemma below identifies conditions under which the system  \eqref{tsys} satisfies Assumption \ref{as:H} {\em without} imposing Assumption \ref{as:g}, that is, avoiding the partition of the coordinates into actuated and un--actuated. Towards this end, the following assumption is needed.

\begin{assumption}
\lab{ass4}
The matrix $A_1$ is {\em diagonally stable} \cite{KASBHA}. That is, there exists $P \in \rea^{n \times n}$, $P=\diag\{p_i\}>0$ such that
\begin{equation} \label{eq:lmi}
PA_1+A_1^\top P =:-2S < 0.
\end{equation}
Moreover, the matrix $A_2$ verifies
\begequ
\lab{moncon}
A_2^\top P \diag\{T_i^3\}+  \diag\{T_i^3\} P A_2  \leq 0.
\endequ
\end{assumption}
\vspace{0.2cm}

Conditions for diagonal stability of a matrix have been studied intensively, see \cite{KASBHA} for a survey. Necessary and sufficient conditions were first reported in \cite{BARetal}---see also \cite{SHOetal} for a simple proof. A sufficient condition, given  in \cite{Farina},  is that it is a Metzler matrix (namely, its non diagonal elements are nonnegative).  

Since $A_2$ is Hurwitz, condition \eqref{moncon} is trivially satisfied if $A_2$ is {\em diagonal}, which is the case in some physical examples.

\begin{lemma} \label{lem3}
If Assumption \ref{ass4} holds  the vector field \eqref{vecfie} satisfies Assumption \ref{as:H} with
\begequ
\lab{hh}
H(x)=\sum_{i=1}^n p_i \phi_i(x_i)+k,
\endequ
where
\begequ
\lab{phii}
\phi_i(x_i)=\frac{1}{5}(x_i+\bar T_i)^5  - \Psi_i(\bar T_i) x_i
\endequ
and
$$
k= -\frac{1}{5}\sum_{i=1}^{n}p_i \bar T_i^5.
$$
\end{lemma}
\begin{proof}
Point (iii) of  Assumption \ref{as:H} is trivially satisfied by \eqref{hh}.

We proceed now to prove point (i). Replacing \eqref{phii} in \eqref{hh} and grouping terms yields
$$
H(x) = \frac{1}{5}\sum_{i=1}^{n}p_i (x_i+\bar T_i)^5 - x^\top P \Psi(\bar T)+k,
$$
Now, notice that
$$
\nabla H(x)=P \Phi(x),
$$
where
\begequ
\lab{phi}
\Phi(x):=\Psi(x+\bar T)-\Psi(\bar T).
\endequ
On the other hand, from \eqref{opelooequ} it follows that the systems vector field may be written as
$$
 f(x)  =   A_1 \Phi(x)+A_2 x.
$$
Consequently,
\begequarrs
[\nabla H(x)]^\top f(x,\theta) & = & \Phi^\top(x)P[ A_1 \Phi(x)+A_2 x ]\\
& = &  -\Phi^\top(x) S  \Phi(x)+ \Phi^\top(x) P A_2 x,
\endequarrs
where we have used \eqref{eq:lmi} to obtain the second identity. Now, condition \eqref{moncon} ensures that the function $h:\rea^n \to \rea^n$
$$
h(x):= A_2^\top P\Psi(x), 
$$
is monotonically decreasing \cite{PAVetal}. That is, for all $a,b \in \rea^n$,
$$
[h(a)-h(b)]^\top (a-b) \leq 0.
$$
Consequently,
$$
\Phi^\top(x) P A_2 x=[h(x+\bar T)-h(\bar T)]^\top x \leq 0 
$$
completing the proof of point (i).

To prove point (ii) we notice that
$$
\tilde f(x)=A_1 \tilde \Phi(x) + A_2 \tilde x,
$$
while
$$
\nabla H(x) - \nabla H(x^*)=P \tilde \Phi(x).
$$
Hence, the claim is established invoking the same arguments used above and defining
$$
Q(x)=\tilde \Phi^\top(x) S \tilde \Phi(x).
$$

Finally, the second derivative of the functions $\phi_i(x_i)$ yields
$$
\phi_i''(x_i)=4 (x_i + \bar T_i)^3=4 T_i^3,
$$
which is non--negative because $T_i \geq 0$. Hence, the functions $\phi_i(x_i)$ are convex as requested by condition (iv) of Assumption \ref{as:H}. This completes the proof.\\
\end{proof}

Direct application of Lemma 1 leads to the following. 
\begin{corollary}
If Assumption \ref{ass4} holds, the thermal system \eqref{tsys} defines a passive map $\tilde u \mapsto e$
with storage function $U(x)$, where 
\begequarrs
e & = & G^\top P \tilde \Phi(x) \\
U(x) & = & H(x) - x^\top P \Phi(x^*) - H(x^*) + (x^*)^\top P \Phi(x^*).
\endequarrs
\end{corollary}
\subsubsection{Robust PI--PBC of the thermal system}
	To present the robust PI--PBC for systems verifying Assumption \ref{as:g} we partition the vector of temperatures into its un--actuated and actuated components
$$
T=\lef[{c} T_u \\ T_a\rig],\;T_u:=\lef[{c} T_1 \\ T_2 \\ \vdots \\ T_{n-m}\rig],\;T_a:=\lef[{c} T_{n-m+1}\\T_{n-m+2}\\ \vdots \\ T_n\rig],
$$
partition $P$ as
\begin{equation*}
	P = \begin{bmatrix} P_1 & \mathbf{0}_{(n-m) \times m} \\ \mathbf{0}_{m\times (n-m)} & D \end{bmatrix},
\end{equation*}
and do the same with the vector function $\Psi(T)$. 

The following proposition is a consequence of Lemma \ref{lem3} and Proposition \ref{pro1}.
\begin{proposition}
\lab{pro2}
Consider the system \eqref{tsys} verifying Assumptions \ref{as:g} and \ref{ass4}. Fix any desired temperature $T^*$ verifying \eqref{tsta} and define the PI--PBC 
$$
\begin{aligned}
		u &= - K_P \tilde {\Psi}_a(T_a) + z \\
		\dot{z} &= - K_I   \tilde \Psi_a(T_a) ,
	\end{aligned}
$$
and the controller gains $K_P$ and  $K_I$ are given by \eqref{gaipi}.  {\em For all} diagonal, positive definite matrices  $\Gamma_P \in \mathbb{R}^{m \times m}$ and $\Gamma_I \in \mathbb{R}^{m \times m}$  all trajectories are bounded and the equilibrium point $(T,z)=(T^*,u^*)$ is {\em globally asymptotically stable}.
\end{proposition}
\begin{proof}
The proof of stability is established invoking item (i) of Proposition 1 and identifying
$$
\tilde \Phi_a(x_a)|_{x_a=T_a - \bar T_a}= \tilde {\Psi}_a(T_a).
$$
To prove asymptotic stability we invoke item (ii) and observe that the augmented error signal \eqref{ea} is given in this case by
\begin{equation*}
		e_a= \begin{bmatrix}   \tilde {\Psi}^\top(T) S   \\  G_2^\top D \end{bmatrix}  \tilde {\Psi}(T) .
\end{equation*}
Since $e_a$ verifies \eqref{limea} and $S$ is positive definite we conclude that $ \tilde {\Psi}(T(t)) \to 0$ and consequently $T(t)\to T^\star$.  
\end{proof}


Physically,  considering  matrix $G$ as \eqref{g}  means that for $m$ heating elements there are $n-m$ measured points that are not directly heated by these elements.



\section{Concluding Remarks}
\lab{sec7}	
In this work we identify a class of nonlinear systems for which it is possible to design robust PI controllers with guaranteed stability properties. The class consists of input affine systems with known, constant input matrix $G$ and $n-m$ zero rows. We assume that only the states associated to the non--zero rows of  $G$ are measurable.The systems  have an open--loop stable equilibrium, but is different from the desired operating point. To handle this situation, we follow \cite{JAYetal} and generate new passive outputs for the incremental model, hence the name PI--PBC. Associated to the open--loop stable equilibrium a  Lyapunov function of the form \eqref{h} is assumed to exist. We underscore that, besides convexity, there is no assumption on the function $H_u(x_u)$, which is unknown. Moreover,  the controller does not require the measurement of $x_u$. The functions $\phi_i(x_i)$ are assumed convex and known, but the coefficient $d_i$ are unknown. Under these conditions, we show that, for a well identified class of PI tuning gains, see \eqref{gaipi}, global stability of the proposed PI--PBC is guaranteed. Conditions that ensure global asymptotic stability, are also derived.  

%
\section*{Acknowledgments}
%

This work was supported by the Ministry of Education and Science of Russian Federation (Project 14.Z50.31.0031).

\bibliographystyle{IEEEtran}
\bibliography{IEEEabrv,bibliography}

\begin{thebibliography}{10}
\providecommand{\url}[1]{#1}
\csname url@samestyle\endcsname
\providecommand{\newblock}{\relax}
\providecommand{\bibinfo}[2]{#2}
\providecommand{\BIBentrySTDinterwordspacing}{\spaceskip=0pt\relax}
\providecommand{\BIBentryALTinterwordstretchfactor}{4}
\providecommand{\BIBentryALTinterwordspacing}{\spaceskip=\fontdimen2\font plus
\BIBentryALTinterwordstretchfactor\fontdimen3\font minus
  \fontdimen4\font\relax}
\providecommand{\BIBforeignlanguage}[2]{{%
\expandafter\ifx\csname l@#1\endcsname\relax
\typeout{** WARNING: IEEEtran.bst: No hyphenation pattern has been}%
\typeout{** loaded for the language `#1'. Using the pattern for}%
\typeout{** the default language instead.}%
\else
\language=\csname l@#1\endcsname
\fi
#2}}
\providecommand{\BIBdecl}{\relax}
\BIBdecl

\bibitem{VAN}
A.~van~der Schaft, \emph{L2-Gain and Passivity Techniques in Nonlinear
  Control}.\hskip 1em plus 0.5em minus 0.4em\relax Springer-Verlag New York,
  Inc., 1996.

\bibitem{JAYetal}
B.~Jayawardhana, R.~Ortega, E.~Garc\'ia-Canseco, and F.~Casta\~nos, ``Passivity
  of nonlinear incremental systems: Application to {PI} stabilization of
  nonlinear {RLC} circuits,'' \emph{Systems \& {C}ontrol {L}etters}, vol.~56,
  no. 9-10, pp. 618--622, 2007.

\bibitem{CASetal}
F.~Casta\~nos, B.~Jayawardhana, R.~Ortega, and E.~Garc\'ia-Canseco,
  ``Proportional plus integral control for  set--point regulation of  a 
  class of nonlinear {RLC} circuits,'' \emph{Circuits, Systems \& Signal
  Processing}, vol.~28, no.~4, pp. 609--623, 2009.

\bibitem{HERetal}
M.~Hernandez-Gomez, R.~Ortega, F.~Lamnabhi-Lagarrigue, and G.~Escobar,
  ``Adaptive {PI} stabilization of switched power converters,'' \emph{Control
  Systems Technology, IEEE Transactions on}, vol.~18, no.~3, pp. 688--698, May
  2010.

\bibitem{TALetal}
R.~J. Talj, R.~Ortega, and M.~Hilairet, ``A controller tuning methodology for
  the air supply system of a {PEM} fuel-cell system with guaranteed stability
  properties,'' \emph{International Journal of Control}, vol.~61, pp.
  1706--1719, 2009.

\bibitem{MARALE}
G.~Marmidis and A.~Alexandridis, ``A passivity-based {PI} control design for
  {DC}-drives,'' in \emph{Control and Automation, 2009. MED '09. 17th
  Mediterranean Conference on}, June 2009, pp. 1511--1516.

\bibitem{SANetal}
J.~L. Meza and V.~Santiba\~{n}ez, ``Analysis via passivity theory of a class of
  nonlinear {PID} global regulators for robot manipulators,'' in \emph{{IASTED}
  international conference on robotics and applications}, October 1999.

\bibitem{DONJUN}
A.~Donaire and S.~Junco, ``On the addition of integral action to
  port--controlled {H}amiltonian systems,'' \emph{Automatica}, vol.~45, no.~8,
  pp. 1910 -- 1916, 2009.

\bibitem{ORTROM}
R.~Ortega and J.~G. Romero, ``Robust integral control of port--{H}amiltonian
  systems: {T}he case of non-passive outputs with unmatched disturbances,''
  \emph{Systems \& Control Letters}, vol.~61, no.~1, pp. 11 -- 17, 2012.

\bibitem{ROMDONORT}
J.~G. Romero, A.~Donaire, and R.~Ortega, ``Robust energy shaping control of
  mechanical systems,'' \emph{Systems \& Control Letters}, vol.~62, no.~9, pp.
  770 -- 780, 2013.

\bibitem{ANTAST}
R.~Antonelli and A.~Astolfi, ``Continuous stirred tank reactors: Easy to
  stabilise?'' \emph{Automatica}, vol.~39, no.~10, pp. 1817 -- 1827, 2003.

\bibitem{ORTKEL}
R.~Ortega and R.~Kelly, ``{PID} {S}elf-{T}uners: {S}ome theoretical and
  practical aspects,'' \emph{Industrial Electronics, {IEEE} Transactions on},
  vol. IE-31, no.~4, pp. 332--338, Nov 1984.

\bibitem{SASBOD}
S.~Sastry and M.~Bodson, \emph{Adaptive Control: {S}tability, {C}onvergence,
  and {R}obustness}.\hskip 1em plus 0.5em minus 0.4em\relax Prentice Hall,
  1994.

\bibitem{LANTIS}
P.~Lancaster and M.~Tismenetsky, \emph{{T}he {T}heory of {M}atrices}.\hskip 1em
  plus 0.5em minus 0.4em\relax Academic Press, 1985.

\bibitem{KHA}
H.~K. Khalil, \emph{Nonlinear Systems}.\hskip 1em plus 0.5em minus 0.4em\relax
  Prentice Hall, 2002.

\bibitem{PAVetal}
A.~Pavlov, A.~Pogromsky, N.~van~de Wouw, and H.~Nijmeijer, ``{C}onvergent
  dynamics, a tribute to {B}oris {P}avlovich {D}emidovich,'' \emph{Systems \&
  Control Letters}, vol.~52, no. 3–4, pp. 257 -- 261, 2004.

\bibitem{ebert2004model}
J.~Ebert, D.~De~Roover, L.~Porter, V.~A. Lisiewicz, S.~Ghosal, R.~Kosut, and
  A.~Emami-Naeini, ``Model--based control of rapid thermal processing for
  semiconductor wafers,'' in \emph{American Control Conference, 2004.
  Proceedings of the 2004}, vol.~5, June 2004, pp. 3910--3921.

\bibitem{schaper1992modeling}
C.~Schaper, Y.~Cho, and T.~Kailath, ``{L}ow-order modeling and dynamic
  characterization of rapid thermal processing,'' \emph{Applied Physics A},
  vol.~54, no.~4, 1992.

\bibitem{KASBHA}
E.~Kaszkurewicz and A.~Bhaya, \emph{{M}atrix {D}iagonal {S}tability in
  {S}ystems and {C}omputation}.\hskip 1em plus 0.5em minus 0.4em\relax
  Birkhauser, 1991.

\bibitem{BARetal}
G.~P. Barker, A.~Berman, and R.~J. Plemmons, ``Positive diagonal solutions to
  the {L}yapunov equations,'' \emph{Linear and Multilinear Algebra}, vol.~5,
  no.~4, pp. 249--256, 1978.

\bibitem{SHOetal}
R.~Shorten, O.~Mason, and C.~King, ``An alternative proof of the {B}arker,
  {B}erman, {P}lemmons ({BBP}) result on diagonal stability and extensions,''
  \emph{{L}inear {A}lgebra and its {A}pplications}, vol. 430, no.~1, pp. 34--
  40, 2009.

\bibitem{Farina}
L.~Farina and S.~Rinaldi, \emph{Positive {L}inear {S}ystems: {T}heory and
  {A}pplications}.\hskip 1em plus 0.5em minus 0.4em\relax {J}ohn {W}iley \&
  {S}ons, {I}nc, 2000.

\end{thebibliography}

\end{document}